\documentclass[a4paper,draft]{article}
\usepackage{mathrsfs}

\usepackage[centertags]{amsmath}
\usepackage{amsfonts}
\usepackage{amssymb}
\usepackage{amsthm}
\usepackage{indentfirst}
\usepackage{cite}

\makeatletter

\newcommand{\Rmnum}[1]{\uppercase\expandafter{\romannumeral #1}}

\theoremstyle{plain}
\newtheorem{thm}{Theorem}[section]

\newtheorem{lemma}[thm]{Lemma}

\theoremstyle{definition}

\theoremstyle{remark}
\newtheorem{rem}{Remark}[section]
\numberwithin{equation}{section}

\begin{document}

\title{\textbf{Darboux Transformation for the Non-isospectral AKNS Hierarchy and Its Asymptotic Property}}

\author{Lingjun Zhou\thanks{\textit{E-mail address:}zhoulj@mail.tongji.edu.cn}
    \medskip\\
    \textit{\small Department of Mathematics, Tongji University, Shanghai, 200092, PR China}
    }
\date{}


\maketitle


\noindent\hrulefill

\noindent\textbf{\large Abstract}{\hfill}
  \smallskip
  {\small

In this article, the Darboux transformation for the non-isospectral
AKNS hierarchy is constructed. We show that the Darboux
transformation for the non-isospectral AKNS hierarchy is not an
auto-B\"acklund transformation, because the integral constants of
the hierarchy will be changed after the transformation. The
transform rule of the integral constants will be also derived. By
this means, the soliton solutions of the nonlinear equations derived
by the non-isospectral AKNS hierarchy can be found.


  }

\noindent\hrulefill


\section{Introduction}

The AKNS hierarchy is one of the most important integrable systems.
Many nonlinear equations, which involve many famous nonlinear
differential equation, are equivalent to the integrability condition
of AKNS hierarchy \cite{GH,GHZ,RS}. A unified approach to construct
Darboux transformations with matrix form for AKNS hierarchy was
founded by C.H.\ Gu \cite{Gu3,Gu,Gu1,Gu2,GZ,GZ1}. This approach
works for many soliton equations \cite{Gu}. The method in these
literatures can also be generalized to the non-isospectral AKNS
hierarchy \cite{Zhoulj}. The Darboux transformation for the
non-isospectral AKNS hierarchy has an essential difference from the
standard case, that its integral constants are not conserved by the
transform. So one cannot get the soliton solution of the relevant
nonlinear equation by acting the Darboux transformation on the seed
solution of its own \cite{TZ,TZ1,Zhoulj}. In this article, one will
see that the relation of the integral constants between the relevant
non-isospectral AKNS hierarchies can be calculated through the
asymptotic property of the elementary solution. Then the soliton
solution of a certain differential equation can be found by acting
the Darboux transformation on the seed solution of another equation.

\section{Non-isospectral AKNS Hierarchy}

The standard AKNS hierarchy is generalized by Ablowitz, Kaup, Newell
and Segur\cite{AKNS} in 1974. Now we use the language in \cite{TU}
to define it. We assume throughout that the matrix $J\in
\mathrm{sl}(N)$ in this article is fixed, diagonal, with distinct
complex eigenvalues,\begin{equation}
J=\mathrm{diag}\,(J_1,\ldots,J_N),\quad J_i\neq J_j
\quad\mathrm{if}\quad i\neq j,\end{equation}and
\begin{align}\mathrm{sl}(N)_J &=\{X\in \mathrm{sl}(N):\;
[J,X]=0\}\\\mathrm{sl}(N)_J^\bot&=\{Y\in \mathrm{sl}(N):\;
\mathrm{tr}(XY)=0\quad\mbox{for }X\in
\mathrm{sl}(N)_J\}\end{align}denote the centralizer of $J$ and its
orthogonal complement in $\mathrm{sl}(N)$ respectively. One can
easily verify the following facts.

\begin{lemma} $sl(N)$ has the direct sum decomposition $sl(N)=sl(N)_J\oplus
sl(N)_J^\bot$ with respect to vector space. \end{lemma}

\begin{lemma} The matrix $P\in sl(N)_J^\bot$ if and only if the diagonal coefficients
of P vanish. \end{lemma}

\begin{lemma} The mapping $\mathrm{ad}\,J:sl(N)\to sl(N)_J^\bot$
is a homomorphism, and $\ker(\mathrm{ad}\,J)=sl(N)_J$, which is
equivalent to that the restriction of the mapping
$\mathrm{ad}\,J:sl(N)_J^\bot\to sl(N)_J^\bot$ is an isomorphism.
\end{lemma}

We now turn to the non-isospectral AKNS hierarchy. Assume
that\begin{equation} U(\lambda)=\lambda J+P,\quad
V(\lambda)=\sum_{i=0}^nV_i\lambda^i,\end{equation} where $\lambda$
satisfies the scalar equation
\begin{equation}\label{As}\lambda_t=\sum_{i=0}^n f_i\lambda^i\end{equation} and $P\in
sl(N)_J^\bot,\; V_i\in sl(N)$ are matrices independent of $\lambda$.
The coupled $N\times N$ matrix equations
\begin{equation}\label{AKNS}
\left\{\begin{aligned}
\Phi_x&=U(\lambda)\Phi\\
\Phi_t&=V(\lambda)\Phi\end{aligned}\right. \end{equation} are called
AKNS system and the system is integrable if and only if the zero
curvature condition
\begin{equation}\label{ZCE} U_t-V_x+[U,V]=0\end{equation} holds.
Comparing the coefficients of $\lambda$ in \eqref{ZCE} leads to the
following equations
\begin{align}\label{AKNSCon}
&\left[J,V_n\right]=0,\\
\label{AKNSCoi}&f_iJ-V_{i,x}+\left[J,V_{i-1}\right]+\left[P,V_i\right]=0
\quad(1\leq i\leq
n),\\\label{AKNSCo0}&f_0J+P_t-V_{0,x}+\left[P,V_0\right]=0.\end{align}
Define \begin{equation} V^{diag}_{i}=\pi_0(V_i),\quad
V^{off}_{i}=\pi_1(V_i), \end{equation} where $\pi_0,\pi_1\in
\mathrm{End}\,(sl(N))$ denote the projection of $sl(N)$ onto
$sl(N)_J$ and $sl(N)_J^\bot$ respectively, then one can find the
following recurrence formulae
\begin{align} &V_n^{off}=(\mathrm{ad}\,J)^{-1}0=0,\\ \label{AKNSVdx}
&V_{i,x}^{diag}=f_iJ+\pi_0\left(\left[P,V_i^{off}\right]\right)\quad(0\leq
i\leq n), \\ \label{AKNSVo}
&V_{i}^{off}=(\mathrm{ad}\,J)^{-1}\left(V_{i+1,x}^{off}-\pi_1\left(\left[P,V_{i+1}\right]\right)\right)\quad(0\leq
i\leq n-1).\end{align} and the nonlinear PDE on $P$\begin{equation}
\label{AKNSPE} P_t-V_{0,x}^{off}+\left[P,V_0^{diag}\right]=0.
\end{equation} If $P$ satisfies the asymptotic condition
\begin{equation}\label{AKNSPC}
\lim_{x\to\infty}|x|^k\partial^m_x(P(x,t))=0 \qquad \mbox{for $t$
and nonnegative integer $k,m$£¬} \end{equation}i.e. $P(t,\cdot)$ is
in the Schwartz class which is denoted by
$\mathscr{S}(\mathbb{R},sl(N)_J^\bot)$, then $V_i$ will be
determined uniformly up to $n$ integral constants $\alpha_i(t)\in
sl(N)$. In fact, $V_i^{diag}(t)$ can be defined by
\begin{equation}\label{AKNSIC}
V_i^{diag}=\alpha_i(t)+f_iJx+\int_{-\infty}^x\pi_0\left(\left[P,V_i^{off}\right]\right)\mathrm{d}x.
\end{equation} Moreover, we have the following conclusion.

\begin{thm} For $t$ and nonnegative integer $k,m$, \begin{equation}\label{AKNSVC}
\lim_{x\to-\infty}|x|^k\partial^m_x(V_i-f_iJx-\alpha_i(t))=0\qquad(0\leq
i\leq n), \end{equation} and especially \begin{equation}
\lim_{x\to-\infty}(V-\sum_{i=0}^nf_iJx\lambda^i)=\sum_{i=0}^n\alpha_i(t)\lambda^i.
\end{equation} \end{thm}

\begin{proof} Assume inductively that \eqref{AKNSVC} holds for $i$, then \begin{equation}\begin{split}
&\lim_{x\to-\infty}|x|^k\partial^m_x([P,V_i])\\
=&\lim_{x\to-\infty}|x|^k\partial^m_x([P,V_i-f_iJx-\alpha_i(t)]+[P,f_iJx+\alpha_i(t)])=0.
\end{split}\end{equation} It follows that \begin{equation}
\lim_{x\to-\infty}|x|^k\partial^m_xV_{i-1}^{off}=(\mathrm{ad}\,J)^{-1}\lim_{x\to-\infty}|x|^k
\partial^m_x\left(V_{i,x}^{off}-\pi_1\left(\left[P,V_i\right]\right)\right)=0.
\end{equation} So $V_{i-1}^{diag}$ can be defined by \begin{equation}
V_{i-1}^{diag}=\alpha_{i-1}(t)+f_{i-1}Jx+\int_{-\infty}^x\pi_0\left(\left[P,V_{i-1}^{off}\right]\right)
\mathrm{d}x,\end{equation} and \eqref{AKNSVC} holds for $i-1$.
\end{proof}

\begin{rem} If the polynomial $f(\lambda)$ degenerates to vanish, the relevant AKNS hierarchy is called
\textit{isospectral}, otherwise we called it
\textit{non-isospectral}. That is to say the non-isospectral case is
a generalization of the standard case.
\end{rem}

\section{Darboux Transformation for the Non-isospectral AKNS
Hierarchy}

The Darboux transformation for the non-isospectral problem with
$2\times 2$ matrix coefficients is constructed by C. Rogers and
W.K.\ Schief \cite{RS}. Their method can be generalize to $N\times
N$ case \cite{Zhoulj}. We states the conclusion here.

\begin{thm}\label{DTThm} Consider the $N\times N$ matrix equation
\begin{equation}\label{eq1} \Phi_t=V(\lambda)\Phi,\qquad  V(\lambda)\in
sl(N)\end{equation}
where\begin{equation}V(\lambda)=\sum^n_{i=0}V_i\lambda^i.\end{equation}
The spectral parameter $\lambda$ satisfies the scalar differential
equation
\begin{equation}\label{spde}\lambda_t=f(\lambda)=\sum_{i=0}^{n+2}f_i\lambda^i \end{equation}
where $f(\lambda)$ is a polynomial of $\lambda$,
the highest power of $\lambda$ in $f(\lambda)$ is not more than
$n+2$. Let $h_1,\ldots,h_N$ be the known vector-valued
eigenfunctions of the equation (\ref{eq1}) corresponding to the
parameters $\lambda_{1},\ldots,\lambda_{N}$ and at least two of them
are different. Set
\begin{equation}S=H\Lambda H^{-1},\quad P(\lambda)=\lambda I-S\end{equation}
where\begin{equation}
H=(h_1,\ldots,h_N),\quad\Lambda=diag(\lambda_1,\ldots,\lambda_N).
\end{equation}
Then, the transformation
\begin{align}
&\Phi'=D(\lambda)\Phi=p(\lambda)P(\lambda)\Phi\label{DM}\\
&V'(\lambda)=P(\lambda)V(\lambda)P^{-1}(\lambda)+\frac{dP(\lambda)}{dt}P^{-1}(\lambda)+\frac
{dp(\lambda)}{dt}p^{-1}(\lambda) \label{DT}\end{align} where
\begin{equation} p(\lambda)^{-N}=\det P(\lambda)=(\lambda-\lambda_1)\cdots(\lambda-\lambda_N)
\end{equation} is a Darboux transformation for (\ref{eq1}), i.e.
the following conditions are satisfied.

(i)$V'(\lambda)\in sl(N)$.

(ii)$V'(\lambda)$ has the same polynomial structure as $V(\lambda)$.

\end{thm}

According to Theorem \ref{DTThm}, $V'(\lambda)$ satisfies the
equation \begin{equation} V'(\lambda)(\lambda I-S)=(\lambda
I-S)V(\lambda)+(f(\lambda)I-S_t)-g(\lambda)(\lambda I-S),
\end{equation} where \begin{equation}
g(\lambda)=\frac1N\sum^N_{i=1}\frac{f(\lambda)-f(\lambda_i)}{\lambda-\lambda_i}=\sum^{n+1}_{i=0}g_i\lambda^i.
\end{equation} Comparing the coefficients of $\lambda$ leads to \begin{equation}\label{VEP}\begin{split}
V'_n&=V_n+(f_{n+1}-g_n)+g_{n+1}S=V_n+f_{n+2}S+f_{n+1}-g_n,\\
V'_{n-j}&=V_{n-j}+V'_{n-j+1}S-SV_{n-j+1}+(f_{n-j+1}-g_{n-j})+g_{n-j+1}S\\
&=V_{n-j}+\sum^j_{k=1}\left[V_{n-j+k},S\right]S^{k-1}+\sum^{j+1}_{k=0}f_{n-j+k+1}S^k-g_{n-j}.
\end{split}\end{equation} Such Darboux transformation works evidently for the non-isospectral
AKNS hierarchy, nevertheless the integral constants of $V'(\lambda)$
are generally different from the ones of $V(\lambda)$. The
asymptotic property of $D$ is necessary to determined the relation
between the two sets of the integral constants. To state the
asymptotic property of $D$, we need part of the scattering theory of
Beals and Coifman \cite{BC}. Let \begin{equation}
\Gamma_J=\{\zeta\in\mathbb{C}:\mathrm{Re}\,(\zeta(J_j-J_k))=0,\,1\leq
j\leq k\leq N\}. \end{equation}

\begin{thm}[Beals-Coifman]\label{BCS} If $P\in\mathscr{S}(\mathbb{R},sl(N)_J^\bot)$, then for
$\lambda\in\mathbb{C}\setminus\Gamma_J$, there exist elementary
solutions $\Phi_l(x,\lambda)$ and $\Phi_r(x,\lambda)$ of the
differential equation \begin{equation} \Phi_x=(\lambda J+P)\Phi,
\end{equation} which satisfy \begin{gather} \Phi_l(x,\lambda)=O(\exp(\lambda
Jx))\qquad x\to-\infty\\\Phi_r(x,\lambda)=O(\exp(\lambda Jx))\qquad
x\to +\infty\end{gather} \end{thm}

Now one can construct a Darboux transformation for non-isospectral
AKNS hierarchy in the following way. Assume that
$J_1>J_2>\cdots>J_N$ without losing generality. Firstly, set
\begin{equation} \lambda_i(0)\in \mathbb{C}\setminus\Gamma_J,\quad
\text{for }1\leq i\leq N,
\end{equation} then there exists an interval $(-t_0,t_0)$, such that
$\lambda_i(t)$ solves \eqref{spde} and
$\lambda_i(t)\in\mathbb{C}\setminus\Gamma_J$. Then the eigenfunction
$h_i(t,x)$ solving the integrable non-isospectral AKNS hierarchy can
be denoted by \begin{equation}
h_i(t,x)=\Phi_l(x,\lambda_i(t))L_i(t)=\Phi_r(x,\lambda_0(t))R_i(t).
\end{equation} According to Theorem \ref{BCS}, one can easily find that \begin{equation} \label{DetH}
\det H=\sum_{\sigma\in
S_N}a_\sigma(t)O\left(\exp\left(\sum_{k=0}^N\lambda_k(0)J_{\sigma(k)}\right)x\right),
\end{equation} where $H=(h_1,\ldots,h_N)$ and $S_N$ denotes permutation group with order
$N$. Let \begin{equation}\begin{split}
m=\min\left\{\sum_{k=0}^N\lambda_k(0)J_{\sigma(k)}:\sigma\in
S_N\right\},\\
M=\max\left\{\sum_{k=0}^N\lambda_k(0)J_{\sigma(k)}:\sigma\in
S_N\right\}, \end{split}\end{equation} then
\begin{equation}\label{DetHC}\begin{split} \det H=O(\exp mx) \qquad
x\to-\infty,\\\det H=O(\exp Mx) \qquad
x\to+\infty.\end{split}\end{equation} if the relevant coefficients
$a_\sigma(t)$ in \eqref{DetH} do not vanish. Secondly, choose
\begin{equation}
\lambda_1(0)=\cdots=\lambda_{k_0}(0)<0,\lambda_{k_0+1}(0)=\cdots=\lambda_N(0)>0,
\end{equation} and carefully select $h_1,\ldots,h_N$ such that
\eqref{DetHC} holds, then we can prove that
\begin{equation}\label{APS}
\lim_{x\to-\infty}S=\lim_{x\to-\infty}H\Lambda H^{-1}=\Lambda,
\qquad\text{for $t\in(-t_0,t_0)$}.
\end{equation} and $\displaystyle{\lim_{x\to+\infty}S}$ is also a diagonal matrix similar to
$\Lambda$, which leads to \begin{equation}\label{APofP}
P'(t,\cdot)=P+[J,S]\in\mathscr{S}(\mathbb{R},sl(N)_J^\bot),\qquad\text{for
$t\in(-t_0,t_0)$}.
\end{equation} (One may admit the conclusion for the moment, and we will
prove it in the next section.) Setting $x\to-\infty$, with the above
property, \eqref{VEP} leads to that
\begin{equation}\label{ICDT}\begin{split}
\alpha'_n(t)&=\alpha_n(t)+f_{n+2}\Lambda+(f_{n+1}-g_n),\\
\alpha'_{n-j}(t)&=\alpha_{n-j}(t)+\sum^{j+1}_{k=0}f_{n-j+k+1}\Lambda^k-g_{n-j}\quad(j\geq0).
\end{split}\end{equation} If we set \begin{equation}\label{BetaDef}
\beta_{j}(\Lambda)=\sum^{n-j+1}_{k=0}f_{j+k+1}\Lambda^k-g_{j}\qquad(0\leq
j\leq n), \end{equation} the solution of the nonlinear PDE via
non-isospectral AKNS hierarchy with the integral constants
$\alpha_j(t)$ can be attained by acting the Darboux transformation
on the seed solution of the equation of the hierarchy with the
integral constants $\alpha_j(t)-\beta_{j}(\Lambda)$.

\section{Proof of the Asymptotic Property}

Firstly, we state an elementary inequality.

\begin{lemma}\label{SInEq} Assume that $x_i,y_i\in\mathbb{R}\,(1\leq i\leq N)$, $x_1\leq x_2\leq\cdots\leq
x_N$ and \begin{equation} M=\max\left\{\sum_{i=1}^N
x_iy_{\sigma(i)}:\sigma\in S_N\right\},\;m=\min\left\{\sum_{i=1}^N
x_iy_{\sigma(i)}:\sigma\in S_N\right\} \end{equation} then
\begin{equation}\begin{aligned}&\sum_{i=1}^N x_iy_{\sigma(i)}=M\quad
\text{if and only if}\quad y_{\sigma(1)}\leq
y_{\sigma(2)}\leq\cdots\leq y_{\sigma(N)},\\&\sum_{i=1}^N
x_iy_{\sigma(i)}=m\quad \text{if and only if}\quad y_{\sigma(1)}\geq
y_{\sigma(2)}\geq\cdots\geq
y_{\sigma(N)}.\end{aligned}\end{equation}
\end{lemma}
Let $X=\{x_1,x_2,\ldots,x_n\},\; Y=\{y_1,y_2\ldots,y_n\}$ be finite
sets which satisfy $\# X=\# Y$, and define
\begin{equation}\begin{aligned} &\langle X,Y\rangle_{max}=\max\left\{\sum_{i=1}^N
x_iy_{\sigma(i)}:\sigma\in S_N\right\},\\&\langle
X,Y\rangle_{min}=\min\left\{\sum_{i=1}^N x_iy_{\sigma(i)}:\sigma\in
S_N\right\}.\end{aligned} \end{equation} Then one can see that the
coefficients of $H$ and its adjoint $H^*$ satisfy \begin{equation}
H_{ij}=O(\exp(\lambda_jJ_ix)),\;
H_{ij}^*=O(\exp(\langle\Lambda\setminus\lambda_j,J\setminus
J_i\rangle_{min}x)),\;x\to-\infty,\end{equation} which implies that
\begin{equation}
H_{ik}H^*_{jk}=O\left(\exp(\lambda_kJ_i+\langle\Lambda\setminus\lambda_k,J\setminus
J_j\rangle_{min})x\right).\end{equation}

Now, we begin to prove \eqref{APS}, which is equivalent to
\begin{equation} \lim_{x\to-\infty}ent_{ij}((\det H)^{-1}H\Lambda H^*)=\delta_{ij}\lambda_j,
\end{equation} where $J=\{J_1,\ldots,J_N\}$ and $\Lambda=\{\lambda_1,\ldots,\lambda_N\}$
are regarded as two sets with $N$ elements. We prove it in the
following cases.

\noindent \textbf{Case 1}, $i<j$. Noting $J_i>J_j$, it implies that
\begin{equation} \lambda_kJ_i+\langle
\Lambda\setminus\lambda_k,J\setminus
J_j\rangle_{min}>\lambda_kJ_j+\langle
\Lambda\setminus\lambda_k,J\setminus J_j\rangle_{min}\geq m,\;
k>k_0.
\end{equation} So it follows from \begin{equation}
\sum_{k=0}^NH_{ik}H^*_{jk}=0,
\end{equation} that \begin{equation}\begin{split}
&\lim_{x\to-\infty}(\det H)^{-1}\left(\sum_{k\leq
k_0}H_{ik}H_{jk}^*\right)=-\lim_{x\to-\infty}\exp(-mx)\left(\sum_{k>k_0}H_{ik}H_{jk}^*\right)\\
=&-\lim_{x\to-\infty}\exp(-mx)\left(\sum_{k>k_0}O(\exp(\lambda_kJ_i+\langle
\Lambda\setminus\lambda_k,J\setminus
J_j\rangle_{min})x)\right)=0\end{split}
\end{equation} and \begin{equation}\begin{split}
&\lim_{x\to-\infty}ent_{ij}S=\lim_{x\to-\infty}(\det H)^{-1}\left(\sum_{k=1}^N\lambda_kH_{ik}H^*_{jk}\right)\\
=&\lim_{x\to-\infty}\exp(-mx)\left(\lambda_1\sum_{k\leq
k_0}H_{ik}H^*_{jk}+\lambda_N\sum_{k>k_0}H_{ik}H^*_{jk}\right)=0.
\end{split}\end{equation}

\noindent \textbf{Case 2}, $i>j$. It is similar to case 1 that
$J_i<J_j$ leads to that
\begin{equation} \lambda_kJ_i+\langle
\Lambda\setminus\lambda_k,J\setminus
J_j\rangle_{min}>\lambda_kJ_j+\langle
\Lambda\setminus\lambda_k,J\setminus J_j\rangle_{min}\geq m,\;k\leq
k_0.
\end{equation} Then one can find that
\begin{equation}\lim_{x\to-\infty}(\det H)^{-1}\left(\sum_{k>
k_0}H_{ik}H_{jk}^*\right)=-\lim_{x\to-\infty}(\det
H)^{-1}\left(\sum_{k\leq k_0}H_{ik}H_{jk}^*\right)=0,
\end{equation} which implies  \begin{equation}\lim_{x\to-\infty}ent_{ij}S=0.
\end{equation}

\noindent \textbf{Case 3}, $i=j\leq k_0$. It follows from Lemma
\ref{SInEq} that \begin{equation} \lambda_kJ_i+\langle
\Lambda\setminus\lambda_k,J\setminus J_i\rangle_{min}\geq m,\quad
k>k_0,
\end{equation} which implies that \begin{equation}\begin{split} &\lim_{x\to-\infty}(\det
H)^{-1}\left(\sum_{k\leq
k_0}H_{ik}H_{ik}^*\right)\\=&\lim_{x\to-\infty}(\det
H)^{-1}\left(\sum_{k=1}^nH_{ik}H_{ik}^*
-\sum_{k>k_0}H_{ik}H_{ik}^*\right)\\=&\lim_{x\to-\infty}(\det
H)^{-1}\left(\det H -\sum_{k>k_0}H_{ik}H_{ik}^*\right)=1.
\end{split}\end{equation} Hence \begin{equation}
\lim_{x\to-\infty}ent_{ii}S=\lim_{x\to-\infty}(\det
H)^{-1}\left(\lambda_1\sum_{k\leq
k_0}H_{ik}H^*_{ik}+\lambda_N\sum_{k>k_0}H_{ik}H^*_{ik}\right)=\lambda_1.
\end{equation}

\noindent \textbf{Case 4}, $i=j>k_0$. It is similar to case 3 that
one can find \begin{equation}
\lim_{x\to-\infty}ent_{ii}S=\lim_{x\to-\infty}(\det
H)^{-1}\left(\lambda_1\sum_{k\leq
k_0}H_{ik}H^*_{ik}+\lambda_N\sum_{k>k_0}H_{ik}H^*_{ik}\right)=\lambda_N.
\end{equation}

That $\displaystyle{\lim_{x\to+\infty}S}$ is a diagonal matrix
similar to $\Lambda$ can be proved by the same means. Combined with
the asymptotic property of the elementary solution $\Phi$, one can
find \eqref{APofP} evidently.

\section{The Soliton Solution of the Non-isospectral MKdV Equation}

The following equation \begin{equation}\label{NISMKdV}
u_t+(1-\frac14x)(u_{xxx}+6u^2u_x)-\frac34u_{xx}-u^3-\frac12u_x\int_{-\infty}^xu^2\mathrm{d}x=0,
\end{equation} compared with the standard MKdV
equation \begin{equation} u_t+6u^2u_x+u_{xxx}=0,
\end{equation} is called the non-isospectral MKdV equation, which is
equivalent to the zero curvature condition of the non-isospectral
AKNS system defined by
\begin{equation*} n=3,J=\mathrm{diag}\,(1,-1),p=-q=u(t,x),\alpha_3=-4J,\alpha_0=\alpha_1=\alpha_2=0.
\end{equation*} and $f(\lambda)=\lambda^3$. Directly calculation leads to
$\lambda^2=(\kappa-2t)^{-1}\,(\kappa\in\mathbb{R})$. Choose
\begin{equation} \Lambda=\left(\begin{array}{cc}
\lambda_0&0\\
0&-\lambda_0\end{array}\right)=\left(\begin{array}{cc}
-(\kappa_0-2t)^{-\frac12}&0\\
0&(\kappa_0-2t)^{-\frac12}\end{array}\right),
\end{equation} then one can find
\begin{equation} g(\lambda)=\lambda^2+\lambda_0^2.\end{equation}
According to \eqref{BetaDef}, $\beta(\Lambda)$ should be defined by
\begin{equation}\begin{array}{ll}
\beta_3(\Lambda)=0,&\beta_2(\Lambda)=f_3-g_2=0,\\
\beta_1(\Lambda)=f_3\Lambda=\Lambda,
&\beta_0(\Lambda)=f_3\Lambda^2-g_0=0.
\end{array}\end{equation} Hence the soliton solution of
\eqref{NISMKdV} can be attained by acting the Darboux transformation
on the nonisospectral AKNS system with the integral constants
$\alpha'(t)=\alpha(t)-\beta(\Lambda)$, i.e.
$\alpha'_3(t)=-4J,\alpha'_2=\alpha'_0=0,\alpha'_1=-\Lambda$. Choose
the trivial solution for the seed solution, i.e. set $P=0$, then the
matrices $U,V_i$ with the integral constants $\alpha'(t)$ are
defined by the following
\begin{equation}\begin{split} U&=\left(\begin{array}{cc}
(\kappa-2t)^{-\frac12}&0\\
0&-(\kappa-2t)^{-\frac12}\end{array}\right)\\
V_3&=\left(\begin{array}{cc}
x-4&0\\0&-x+4\end{array}\right),\qquad V_2=0,\\
V_1&=-\Lambda=\left(\begin{array}{cc}
(\kappa_0-2t)^{-\frac12}&0\\
0&-(\kappa_0-2t)^{-\frac12}\end{array}\right),\qquad V_0=0,
\end{split}\end{equation} and the relevant elementary solution is \begin{equation}
\Phi(x,t,\lambda)=\left(\begin{array}{cc} C(t)\exp(\lambda
x)&0\\0&(C(t))^{-1}\exp(-\lambda x)\end{array}\right),\end{equation}
where \begin{equation}
C(t)=\exp\left(-\int_0^t(4\lambda^3+\lambda_0\lambda)\mathrm{d}t\right).
\end{equation} Set \begin{equation}H=\left(\begin{array}{cc}
C_0(t)\exp(\lambda_0x)&-(C_0(t))^{-1}\exp(-\lambda_0x)\\
(C_0(t))^{-1}\exp(-\lambda_0x)&C_0(t)\exp(\lambda_0x)\end{array}\right),\end{equation}
where \begin{equation}
C_0(t)=\exp\left(-\int_0^t(4\lambda_0^3+\lambda_0^2)\mathrm{d}t\right)=\exp\left(-4\lambda_0-\ln(-\lambda_0)+c_0\right)
\end{equation} and $c_0=(4\lambda_0+\ln(-\lambda_0))|_{t=0}$, then \begin{equation}
S=H\Lambda H^{-1}=\lambda_0\left(\begin{array}{lr} \tanh2\xi&\mathrm
{sech}\,2\xi\\ \mathrm {sech}\,2\xi& -\tanh2\xi\end{array}\right),
\end{equation} where
$\xi=\lambda_0x-4\lambda_0-\ln(-\lambda_0)+c_0$. Following from
$P'=P+[J,S]$, one can see the 1-soliton solution of \eqref{NISMKdV}
\begin{equation} u=2\lambda_0\,\mathrm{sech}\,2\zeta.
\end{equation} By this means, the 2-soliton solution can also be
attained.

\section*{Acknowledgements}

The author would like to thank Prof.\ C. H. Gu, Prof.\ H. S. Hu and
Prof.\ Z. X. Zhou for many helpful discussions and great
encouragement. This work is supported by Program for Young Excellent
Talents in Tongji University.


\begin{thebibliography}{References}

\bibitem{AKNS}Ablowitz, M.J. Kaup, D.J. Newell, A.C. and Segur,
H.\ Studies in Appl.\ Math.\ 53(1974), no.\ 4, 249.

\bibitem{BC}Beals, R. and Coifman, R. R\@. Comm.\ Pure Appl.\ Math. 37-1(1984)39.

\bibitem{Gu3} Gu, C. H\@. Integrable system, Nankai Lectures on Math.\
Phys. World Scientific, Singapore, (1989)162.

\bibitem{Gu} Gu, C. H\@. Analyse, Vari\'et\'es et
Physique, Proc. of Colloque International en l'honneur d'Yvonne
Choquet-Bruhat, Kluwer, (1992).

\bibitem{Gu1}Gu, C. H\@. Lett.\ Math.\ Phys.\ 26(1992)199.

\bibitem{Gu2}Gu, C. H\@. Proc. of
the Workshop on Qualitative Aspect and Applications of Nonlinear
Equations, Trieste, 1993. World Scientific, (1994)11.

\bibitem{GH}Gu, C. H\@. Soliton Theary and Its
Applications. ed. Springer-Verlag and Zhejiang Science and
Technology Publising House, (1995).

\bibitem{GHZ}Gu, C. H., Hu, H. S. and Zhou, Z. X\@.
Darboux transformation in soliton theory and its geometric
applications. Shanghai Scientific and Technical Publishers, (1999).

\bibitem{GZ}Gu, C. H. and  Zhou, Z. X\@.  Lett.\
Math.\ Phys.\ 13(1987)179.

\bibitem{GZ1}Gu, C. H. and  Zhou, Z. X\@. Lett.\ Math.\ Phys.\ 32(1994)1.

\bibitem{RS}Rogers, C. and  Schief, W. K\@. B\"acklund and Darboux
transformations. Cambridge University Press (2002).

\bibitem{TU}Terng, C.-L. and Uhlenbeck, K\@. Communications on Pure and Applied Mathematics. LIII(2000)1.

\bibitem{TZ}Tian, C. and  Zhang, Y. J\@. J. Math.\ Phys.\ 31(1990)2150.

\bibitem{TZ1}Tian, C. and  Zhang, Y. J\@. J. Phys.\ A. 23(1990)2867.

\bibitem{Zhoulj}Zhou, L. J\@. Phys.\ Lett.\ A. 345(2005)314.

\end{thebibliography}
\end{document}